\documentclass[journal,onecolumn]{IEEEtran}

\usepackage{amsmath,amssymb,amsfonts,amsthm}
\usepackage{inputenc}
\usepackage{enumerate}
\usepackage{graphicx}
\usepackage{multicol}
\newtheorem{thm}{Theorem}

\newtheorem{lem}{Lemma}
\newtheorem{cor}{Corollary}

\newtheorem{exam}{Example}

\newtheorem{rem}{Remark}

\def\0{{\mathbf 0}}

\newcommand{\F}{\mathbb{F}}
\newcommand{\Z}{\mathbb{Z}}

\begin{document}

\sloppy

\title{MDS linear codes with one dimensional hull }

\author
{
{Lin Sok\thanks{This research work is supported by Anhui Provincial Natural Science Foundation with grant number 1908085MA04.

Lin Sok is with School of Mathematical Sciences, Anhui University,  230601 Anhui,  P. R. China,
(email: soklin\_heng@yahoo.com). }
}
}

\date{}



\maketitle


\begin{abstract}
We define the Euclidean hull of a linear code $C$ as the intersection of $C$ and its Euclidean dual $C^\perp$. The hull with low dimensions gets much interest due to its crucial role in determining the complexity of algorithms for computing the automorphism group of a linear code and checking permutation equivalence of two linear codes. It has been recently proved that any $q$-ary $[n,k]$ linear code with $q>3$ gives rise to a linear code with the same parameters and having zero dimensional Euclidean hull, which is known as a linear complementary dual code.
This paper aims to explore explicit constructions of families of MDS linear codes with one dimensional Euclidean hull. We obtain several classes of such codes.

\end{abstract}

\section{Introduction}
MDS codes form an optimal family of classical codes. They are closely related to combinatorial designs \cite[p. 328]{MacSlo}, and finite geometries \cite[p. 326]{MacSlo}. Due to their largest error correcting capability for given length and dimension, MDS codes are of great interest in both theory and practice. The most well-known family of MDS linear codes is that of Reed-Solomon codes. MDS linear codes exist in a very restrict condition on their lengths as the famous MDS conjecture states: for every linear
$[n, k, n -k + 1]$ MDS code over $\F_q,$ if $1 < k < q,$ then $n \le q + 1,$ except when $q$ is even and $k = 3 $ or $k = q -1,$
in which cases $n \le q + 2.$ The conjecture was proved by Ball \cite{Ball} for $q$ a prime. However, for some special classes of linear codes, the conjecture may not be true.

The hull of a linear code was studied by Assmus {\em et al.} \cite{AssKey} to classify finite projective planes. In coding theory, the hull of a linear code plays a crucial role in determining the complexity of algorithms for computing the automorphism group of a linear code \cite{Leon82} and for checking permutation equivalence of two linear codes \cite{Leon91,Sendrier00}. In general, the algorithms have been proved to be very effective if the size of the hull is small. We can summarize the work on the hulls of linear codes as follows. Sendrier \cite{Sendrier97} determined the expected dimension of the hull of a random $[n,k]$ code when $n$ and $k$ go to infinity. Skersys \cite{Ske} gave the average dimension of the hulls of cyclic codes. Sangwisut {\em et al. }\cite{SangJitLingUdom} gave enumerations of cyclic codes and negacyclic codes of length $n$ with hulls of a given dimension.

There have been a lot of research on constructions of two types of hulls: the one with dimension zero known as linear complementary dual (LCD) code and that with dimension half of the code length known as self-dual code. We refer to \cite{Gue,GraGul,JinXin,Yan,FangFu,TongWang} for the work on families of MDS self-dual codes and \cite{CarGui,CarGunOzbOzkSol,CarMesTanQiPel18,CarMesTanQi19,
CarMesTanQi18-2,CarMesTanQi19-2,ChenLiu,
Jin,JinBee,LiDingLi,LiLiDingLiu,Massey, MesTanQi,ShiYueYan,YanLiuLiYang} for the LCD codes.

However, there have been very little work on constructing linear codes with other hull dimensions mentioned above except \cite{CarLiMes,LiZeng}. In this work, we will consider the constructions of MDS linear codes with one dimensional hull. We use tools from algbraic function fields in one variable to study such codes. Sufficient conditions for a code to have one dimensional hull are given, and we explicitly construct many families of MDS linear codes with one dimensional hull. Those families are contained in the following theorem.
\begin{thm}
Assume that $q=p^m, m\ge 1,$ is a prime power.
\begin{enumerate}[I.]
\item If $q$ is even and $n\le q-2,$ then there exists a $1$-$d$-hull MDS code with parameters $[n,n-s-1]$ for any $1\le s\le n-3$ (see Corollary \ref{cor:even q}).

\item If $q>5$ is odd, 

\begin{enumerate}[i.]
\item $N\le q-2$, and set
$$K=
\begin{cases}
N-2s-1\text{ if $N$ is even and }1\le s\le \frac{N}{2}-2, \\
N-2s-2\text{ if $N$ is odd and }0 \le s\le \frac{N+1}{2}-3,\\
\end{cases}
$$ 
then there exist $1$-$d$-hull MDS codes with parameters $[N,K]$ and $[N,N-K]$ with the following conditions (see Corollary \ref{cor:square h}): 

	\begin{enumerate}[(1)]
\item $p|N$, $(N-1)|(q-1),$ $N$ even, 
\item $q$  square, $(N-1)|(q-1)$, $N$ even,

\item $ N|\frac{(q-1)}{2}$, 
		\begin{enumerate}[a)]
		\item $N$  even,
		\item $q\equiv 1\pmod 4$, $N$ odd,
		\end{enumerate}

\item $1\le r< m$, $\frac{n(p^r+1)}{2(p^r-1)}$ odd, $N=(t+1)n$, $n=\frac{q-1}{p^r+1}$ a square, $t$ odd, $1\le t\le \lfloor \frac{q-2}{n} \rfloor -1,$
\item $1\le r< m$, $\frac{n(p^r+1)}{2(p^r-1)}$ even, $N=(t+1)n$, $n=\frac{q-1}{p^r+1}$ a square, $1\le t\le \lfloor \frac{q-2}{n} \rfloor -1,$
\item $1\le r< m$, $N=(t+1)n$, $n=\frac{q-1}{p^r-1}$ a square, $r|\frac{m}{2}$, $1\le t\le p^r-2,$
\item $m=2m_0,q_0=p^{m_0}$, $t$ even, $1\le t< q_0$, $N=q_0t,$
\item  $m=2m_0,q_0=p^{m_0}$, $t$ odd, $1\le t< q_0$ and $1\le t\le q_0-1$ for $q=9$, $N=q_0t+1,$

\item $m$ even, $r=p^{m_0},{m_0}|\frac{m}{2},N=2tr^\ell,0\le \ell<m/{m_0},1\le t\le \min((r-1)/2 ,\lfloor \frac{q-2}{2r^\ell} \rfloor),$\\
\item $q\equiv 1\pmod 4, N=2p^\ell,0< \ell <m,$\\
\item $m$ even, $r=p^{m_0},{m_0}|\frac{m}{2},N=(2t+1)r^\ell+1,0\le \ell <m/{m_0},0 \le t \le \min((r-1)/2 ,\lfloor \frac{q-2-r^\ell}{2r^\ell} \rfloor),$\\
\item $q\equiv 1 \pmod 4 ,N=p^\ell+1,0< \ell< m,$\\
	\end{enumerate}

\item $n\le q-1,$ $p|n$ and $(n-1)|(q-1),$ then there exist  $1$-$d$-hull MDS codes with parameters $[n,n-2s+1]$ and  $[n,2s-1]$ for $1\le s\le \lfloor \frac{n}{2}\rfloor$ (see Corollary \ref{con00:gen1}),
\item $1\le r<m$ and $r|m$, then there exist $1$-$d$-hull MDS codes with parameters $[p^r,p^r-2s+1]$ and  $[p^r,2s-1]$ for $1\le s\le \frac{p^r-1}{2}$ (see Corollary \ref{con0:gen1}),

\item $n\le (q-2)$, then there exist $1$-$d$-hull MDS codes with parameters  as follows (see Corollary \ref{con1:gen1}):
\begin{enumerate}
\item  $[n,n-2s+1]$ and $[n,2s-1]$ for $1\le s\le \lfloor (n-1)/2\rfloor,$
\item  $[n,n-2s]$ and $[n,2s]$ for $n$ even and $1\le s\le n/2-1$,
\end{enumerate}

\item $m\ge 2$, $ r\le m-1$, $r|m$, and set $N=(t+1)p^r$ with $\gcd (p,t+1)=1$ and $1\le t\le  \lfloor \frac{q-1-p^r}{2p^r} \rfloor$, then there exist $1$-$d$-hull MDS codes with parameters $[N,N-2s+1]$ and $[N,2s-1]$ for $1\le s\le \lfloor p^r/2\rfloor $ (see Corollary \ref{con3:gen1}),

\item $n|(q-1)$, and set $N=(t+1)n$ with $1\le t\le \lfloor \frac{q-n-2}{2n}\rfloor,$ then there exist $1$-$d$-hull MDS codes with parameters  as follows (see Corollary \ref{con2:gen1}):
\begin{enumerate}[(1)]
\item  $[N,n-2s+1]$ and $[N,tn+2s-1]$ for $1\le s\le \lfloor (n-1)/2\rfloor$, $p\not|(t+1),$
\item  $[N,2n-2s+1]$ and $[N,(t-1)n+2s-1]$ for $1\le s\le \lfloor (2n-1)/2\rfloor$, $p|(t+1),$
\item  $[N,n-2s]$ and $[N,tn+2s]$ for $n$ even, $1\le s\le \lfloor (n-1)/2\rfloor$, $p\not|(t+1),$
\item  $[N,2n-2s]$ and $[N,(t-1)n+2s]$ for $n$ even, $1\le s\le \lfloor (2n-1)/2\rfloor$, $p|(t+1),$
\item  $[N,N-2s+1]$ and $[N,2s-1]$ for $1\le s\le \lfloor (n-1)/2\rfloor$, $p\not|(t+1),$
\item  $[N,N-2s+1]$ and $[N,2s-1]$ for $1\le s\le \lfloor (2n-1)/2\rfloor$, $p|(t+1)$,
\item  $[N,N-2s]$ and $[N,2s]$ for $n$ even, $1\le s\le \lfloor (n-1)/2\rfloor$, $p\not|(t+1),$
\item  $[N,N-2s]$ and $[N,2s]$ for $n$ even, $1\le s\le \lfloor (2n-1)/2\rfloor$, $p|(t+1).$
\end{enumerate}

\end{enumerate}
\end{enumerate}
\end{thm}

The paper is organized as follows: Section \ref{section:pre} gives preliminaries and background on algebraic geometry (AG) codes. Section \ref{section:characterization} provides a characterization of a genus zero AG code to have one dimensional hull. Section \ref{section:constructions} gives some methods to construct MDS linear codes with one dimensional hull. We give a concluding remark in Section \ref{sec:conclusion}.

\section{Preliminaries}\label{section:pre}

Let $\F_q$ be the finite field with $q$ elements. A 
linear code of length $n$ and dimension $k$ over ${{\mathbb F}_q},$ denoted as $q$-ary $[n,k]$ code, is a $k$-dimensional subspace
of  $ {\mathbb F}_q^n$. The
(Hamming) weight wt$({\bf{x}})$ of a vector ${\bf{x}}=(x_1, \dots,
x_n)$ is the number of nonzero coordinates in it. The {\em minimum
distance ({\rm{or}} minimum weight) $d(C)$} of $C$ is
$d(C):=\min\{{\mbox{wt}}({\bf{x}})~|~ {\bf{x}} \in C, {\bf{x}} \ne
{\bf{0}} \}$. The parameters of an $[n,k]$ code with minimum distance $d$ are written $[n,k,d]$. 
If $C$ is an $[n,k,d]$ code, then from the Singleton bound, its minimum distance is bounded above by
$$d\le n-k+1.$$
A code meeting the above bound is called {\em Maximum Distance Separable} ({MDS}).
The {\em Euclidean inner product} of ${\bf{x}}=(x_1,
\dots, x_n)$ and ${\bf{y}}=(y_1, \dots, y_n)$ in ${\mathbb F}_q^n$ is
${\bf{x}}\cdot{\bf{y}}=\sum_{i=1}^n x_i y_i$. The {\em dual} of $C$,
denoted by $C^{\perp}$, is the set of vectors orthogonal to every
codeword of $C$ under the Euclidean inner product. The (Euclidean) {\em hull} of a linear code $C$ is defined as

$$hull(C)=C\cap C^\perp. $$

If the hull of a linear code $C$ has dimension $s$, then we call $C$ an $s$-$d$-hull code.  With this definition, a linear complementary dual (LCD) code is a $0$-$d$-hull code and a self-dual code of length $n$ is a $\frac{n}{2}$-$d$-hull code.

We refer to Stichtenoth \cite{Stich} for undefined terms related to algebraic function fields. We denote the rational function field of one variable $\F_q(x)$ by $\cal {F}$ and the set of places of $\F_q(x)$ by ${\cal X}$. For an element $\alpha\in \F_q,$ let $P_{\alpha}$ denote the zero place of $x-\alpha$ and $P_\infty$ its pole place. A divisor $G$ of $\cal {F}$ is a formal sum $\sum\limits_{P\in {\cal X}}n_PP$ with only finitely many nonzeros $n_P\in \Z$. The support of $G$ is defined as $supp(G):=\{P|n_P\not=0\}$. The degree of $G$ is defined by $\deg(G):=\sum\limits_{P\in {\cal X}}n_P\deg(P)$. 
For two divisors $G=\sum\limits_{P\in {\cal X}}n_PP$ and $H=\sum\limits_{P\in {\cal X}}m_PP$,  we say that $G\ge H$ if $n_P\ge m_P$ for all places $P\in {\cal X}$.

It is well-known that a nonzero polynomial $f(x)\in \F_q(x)$ can be factorized into irreducible factors as $f(x)=\alpha \prod\limits_{i=1}^s p_i(x)^{e_i},$ with $\alpha\in \F_q^*.$ Moreover, any irreducible polynomial $p_i(x)$ corresponds to a place, say $P_i$. We define the valuation of $f$ at $P_i$ as $v_{P_i}(f):=t$ if $p_i(x)^t|f(x)$ but $p_i(x)^{(t+1)}\not|f(x).$ 

For $f\in\F_q(x)$, we define
$$
\begin{array}{c}
(f)_0:=\sum\limits_{P\in Z(f)}v_{P}(f)P, \text{ the zero divisor of } f,\\
(f)_\infty:=\sum\limits_{P\in N(f)}-v_{P}(f)P, \text{ the pole divisor of } f,\\
(f):=(f)_0-(f)_\infty,\text{ the principal divisor of } f,\\
\end{array}
$$
where $Z(f)$ and  $N(f)$ denotes the set of zeros and poles of $f$, respectively. Hence, $$(f)=\sum\limits_{P\in {\cal X}}v_P(f)P.$$ It is well known that the principal divisor has degree $0.$

We say that two divisors $G$ and $H$ are equivalent if $G=H+(y)$ for some rational function $y\in \F_q(x).$
For a divisor $G$, we define
$${\cal L}(G):=\{f\in \F_q(x)\setminus \{0\} |(f)+G\ge 0\}\cup \{0\},$$ 
and 
$${\Omega}(G):=\{\omega\in \Omega_{\cal F}\setminus \{0\} |(\omega)-G\ge 0\}\cup \{0\},$$
where $\Omega_{\cal F}:=\{fdx|f\in \cal {F}\}$, the set of differential forms. 

The dimension of ${\cal L}(G)$ is denoted by $\ell (G),$ and is determined by Riemann-Roch's theorem as follows.

\begin{thm}\cite[Theorem 1.5.15 (Riemann-Roch)]{Stich} Let $W$ be a canonical divisor. Then, for each divisor $G$, the following holds:
$$\ell (G) = \deg G + 1 -g + \ell(W-G),$$
where $g$ is the genus of the smooth algebraic curve.
\end{thm}

For a special divisor $G$, we can determine the dimension of the space ${\cal L}(G)$ as follows.
\begin{lem}\cite[Corollary 1.4.12]{Stich} \label{prop:dim-principal}Assume that a divisor $G$ has degree zero. Then $G$ is principal
if and only if $\ell (G)=1.$
\end{lem}

For any place $P\not=P_\infty,$ let $ v_P{(fdx)}:=v_P{(f)}$ and $ v_{P_\infty}{(fdx)}:=v_{P_\infty}{(f)}-2.$ For an element $\alpha\in \F_q$ and $f\in \cal {F}$ with $v_{P_\alpha}(f)\ge -1$, it is well known that $f(x)$ can be expanded in the neighborhood of $\alpha$ as follows
$$f(x)=\cdots+\frac{a_{-1}}{x-\alpha}+a_0+a_1(x-\alpha)+\cdots .$$
If $f(x)$ is in the above form, the residue $Res_{P_\alpha}(fdx)$ of $fdx$ at $P_\alpha$ is defined to be $Res_{P_\alpha}(fdx):=a_{-1}$.

Over $\F_q(x)$, there are $q+1$ places of degree one, that is, the zero places $P_1,\hdots,P_q$ and the pole place $P_\infty $ (see \cite[Proposition 1.2.1]{Stich}).

Through out the paper, we let $D=P_1+\cdots+P_n$, called the rational divisor, where $P_i:=P_{\alpha_i},\alpha_i \in \F_q,$ for ${1\le i \le n},$ are places of degree one. 

For $G$ a divisor with $supp(D)\cap supp(G)=\emptyset$, define the algebraic geometry code by
$$
C_{\cal L}(D,G):=\{(f(P_1),\hdots,f(P_n))|f\in {\cal L}(G)\},
$$
and the differential algebraic geometry code as
$$
C_{\Omega}(D,G):=\{(Res_{P_1}(\omega),\hdots,Res_{P_n}(\omega))|\omega\in {\Omega}(G-D)\}.
$$

The parameters of an algebraic geometry code $C_{\cal L}(D,G)$ is given as follows.
\begin{thm}\cite[Corollary 2.2.3]{Stich}\label{thm:distance} Assume that $2g -2 < deg(G) < n.$ Then the code $C_{\cal L}(D,G)$  has parameters $[n,k,d]$ satisfying
\begin{equation}
k=\deg (G)-g+1\text{ and } d\ge n-\deg (G).
\label{eq:distance}
\end{equation}

\end{thm}
For $g=0$, from (\ref{eq:distance}), we get $k=\deg (G)+1$ and $d\ge n-\deg (G)$ and thus the Singleton bound holds with equality and the code $C_{\cal L}(D,G)$ is MDS.

 For ${\bf a}=(\alpha_1,\hdots,\alpha_n),{\bf v}=(v_1,\hdots,v_n)\in \F_q^n$ such that $\alpha_1,\hdots,\alpha_n$ are all distinct, and $v_1,\hdots,v_n$ are all nonzero, it is well known that the generalized Reed-Solomon code defined by 
$$
\begin{array}{ll}
GRS_k({\bf a},{\bf v}):=\{(v_1f(\alpha_1),\hdots, v_nf(\alpha_n))|f(x)\in \F_q(x),\deg{f}\le k-1\}&\\
\end{array}
$$
 is an MDS code. Furthermore, it is shown 
in \cite[Proposition 2.3.3]{Stich}, that any algebraic geometry code $C_{\cal L}(D,G)$ with $\deg (G)=k-1$ is equal to the generalized Reed-Solomon code $GRS_k({\bf a},{\bf v})$ defined above. 

Moreover, their parameters are related as follows. For all $1\le i \le n,$\\

$
\begin{cases}
\alpha_i=x(P_i),\\
v_i=u(P_i)
\text{ for some $u(x)\in \F_q(x)$ satisfying }\\
  (u)=(k-1)P_{\infty}-G .\\
\end{cases}
$\\

For $0\le j\le k-1$, the vectors
$$(ux^j(P_1),\hdots,ux^j(P_{n}))=(v_1\alpha_1^j,\hdots,v_{n}\alpha_{n}^j)$$ constitute a basis of $C_{\cal L}(D,G)$, and thus, a generator matrix of $C_{\cal L}(D,G)$ can be expressed as

\begin{equation*}
\left(
\begin{array}{cccc}
v_1&v_2&\hdots&v_{n}\\
v_1\alpha_1&v_2\alpha_2&\cdots&v_{n}\alpha_{n}\\
\vdots&\vdots&\cdots&\vdots\\
v_1\alpha_1^{k-2}&v_2\alpha_2^{k-2}&\cdots&v_{n}\alpha_{n}^{k-2}\\
v_1\alpha_1^{k-1}&v_2\alpha_2^{k-1}&\cdots&v_{n}\alpha_{n}^{k-1}\\
\end{array}
\right).
\end{equation*}

Equivalence of two algebraic geometry codes is characterized through the associated divisors as follows.
\begin{lem}\cite[Proposition 2.2.14]{Stich}\label{prop:stich2} Assume that two divisors $G$ and $H$ are equivalent. Then ${\cal L}(G)$ and  ${\cal L}(H)$ are isomorphic as vector spaces. Moreover, the codes $C_{\cal L}(D,G)$ and $C_{\cal L}(D,H)$ are equivalent.
\end{lem}

The dual of the algebraic geometry code $C_{\cal L}(D,G)$ can be described as follows.

\begin{lem}\cite[Theorem 2.2.8]{Stich}\label{lem:01} With the above notation, the two codes $C_{\cal L}(D,G)$ and $C_{\Omega}(D,G)$ are dual to each other.
\end{lem}

Moreover, the differential code $C_{\Omega}(D,G)$ is determined as follows.
\begin{lem}\cite[Proposition 2.2.10]{Stich}\label{lem:1} With the above notation,
$C_{\Omega}(D,G)=C_{\cal L}(D,D-G+(\omega))$ for some differential function $\omega$ satisfying $v_{P_i}(\omega)=-1$ and $Res_{P_i}(\omega)=1$ for $1\le i \le n.$ 
\end{lem}

For simplicity, we let, in the sequel, 
$ h(x)=\prod\limits_{i=1}^n(x-\alpha_i)$ and $h'(x)=\frac{dh}{dx}, \text{ the derivative of }h$ with respect to $x$. Then 
$$h'(x)=\sum\limits_{i=1}^n\prod\limits_{j=1,j\not= i}^n(x-\alpha_j).$$
A simple calculation gives $$\omega_{h'}:=\frac{h'}{h}dx=\left(\frac{1}{x-\alpha_1}+\cdots+\frac{1}{x-\alpha_n}\right)dx, $$ and hence, the divisor $(\omega_{h'})$ of $\omega_{h'}$ 
$$(\omega_{h'})=(h')-D+(n-2)P_{\infty}$$
satisfies
$$
v_{P_i}(\omega_{h'})=-1,
Res_{P_{i}}(\omega_{h'})=1,\forall 1\le i\le n .
$$

\section{Characterization of one dimensional hull}\label{section:characterization}
The following lemma gives sufficient conditions for two algebraic geometry codes to have one dimensional intersection subcode.
\begin{lem}\label{lem:characterization}
 \text{With the same notation as above, assume that}
\begin{enumerate}[(i)]
\item $A,B,G$ are divisors such that $A\ge 0,B\ge 0$ and $supp(A),supp(B),supp(D)$ and $supp(G)$ are pairwise disjoint,
\item $\deg G<n$,
\item $ G-A-B$ is a principal divisor.
\end{enumerate}
Then $C_{\cal L}(D,G-A+(z))\cap C_{\cal L}(D,G-B+(z))$ is a one-dimensional code for any $z\in \F_q(x)$ satisfying $v_{P_i}(z)=0$ for all $1\le i\le n.$
\end{lem}
\begin{proof} Let $z\in \F_q(x)$ satisfying $v_{P_i}(z)=0$ for all $1\le i\le n.$ Let $c\in C_{\cal L}(D,G-A+(z))\cap C_{\cal L}(D,G-B+(z)).$ Then 
$c=(f_1(P_1)\hdots,f_1(P_n))=(f_2(P_1)\hdots,f_2(P_n))$ for some $f_1\in {\cal L}(G-A+(z)),f_2\in {\cal L}(G-B+(z))$ with $(f_1)+G-A+(z)\ge 0$ and $(f_2)+G-B+(z)\ge 0$. Since $A,B$ are positive divisors, $(f_1z)+G\ge 0$ and $(f_2z)+G\ge 0$. Thus $f_1z-f_2z\in {\cal L}(G)$ and $(f_1z-f_2z)+G\ge 0.$ Since $(P_i)_{1\le i\le n}$ are zeros of $f_1-f_2$ but not of $z$ and $supp(G)\cap supp(D)$ are disjoint, we get $(f_1z-f_2z)+G-D\ge 0.$ Hence $f_1z-f_2z\in {\cal L}(G-D)=\{0\}.$ We now have $f_1=f_2\in {\cal L}(G-A)\cap {\cal L}(G-B)={\cal L}(G-A-B)$, and by Lemma \ref{prop:dim-principal}, $\ell (G-A-B)=1$. We conclude that $\dim (C_{\cal L}(D,G-A+(z))\cap C_{\cal L}(D,G-B+(z)))=1$.
\end{proof}
The following lemma gives sufficient conditions for an algebraic geometry code to have one dimensional hull.
\begin{lem}\label{lem:generalized}
 \text{With the same notation as above, assume that}
\begin{enumerate}[(i)]
\item $A,B,G$ are divisors such that $A\ge 0,B\ge 0$ and $supp(A),supp(B),supp(D)$ and $supp(G)$ are pairwise disjoint,
\item $\deg G<n$,
\item $ G-A-B$ is a principal divisor,
\item the following condition
\begin{equation}\label{eq:char2}2G-A-B-(h')-(n-2)P_{\infty}=(y)
\end{equation} holds for some rational function $y$ such that $y(P_i),\forall 1\le i \le n,$ are squares in $\F_q^*.$
\end{enumerate}
Then $C_{\cal L}(D,G-A+(z))$ is a $1$-$d$-hull code for any $z\in \F_q(x)$ satisfying $(z^2y)(P_i)=1$ for all $1\le i\le n.$
\end{lem}
\begin{proof} Under Condition $(i)$ and the assumption $(z^2y)(P_i)=1$, the code $C_{\cal L}(D,G-A+(z))$ is well defined. From the fact that $C_{\cal L}^\perp (D,G-A+(z))=C_\Omega (D,G-A+(z))=C_{\cal L} (D,D-G+A-(z)+(\omega_{h'}))=C_{\cal L} (D,D-G+A-(z)+(h')-D+(n-2)P_\infty)=C_{\cal L} (D,-G+A-(z)+(h')+(n-2)P_\infty).$ The condition iii) implies that $G-B-(y)-(z)=-G+A+(h')+(n-2)P_{\infty}$. Thus $C_{\cal L}^\perp (D,G-A+(z))=C_{\cal L} (D,G-B-(y)-(z))=C_{\cal L}(G-B+(z)-(z^2y))$

From Lemma \ref{prop:stich2}, ${\cal L}(G-B+(z))$ and $ {\cal L}(G-B+(z)-(z^2y))$ are isomorphic vector spaces. Define $$\phi:C_{\cal L}(D,G-B+(z)-(z^2y))\rightarrow C_{\cal L}(D,G-B+(z))$$ such that $\phi(f(P_1),\hdots,f(P_n))=(\frac{f}{z^2y}(P_1),\hdots,\frac{f}{z^2y}(P_n)).$
Under the condition $(z^2y)(P_i)=1$ for $1\le i\le n$,  we get $\phi(f(P_1),\hdots,f(P_n))=({f}(P_1),\hdots,{f}(P_n))$, and thus
$C_{\cal L}(D,G-B+(z)-(z^2y))=C_{\cal L}(D,G-B+(z))=C_{\cal L}^\perp (D,G-A+(z))$. The rest follows from Lemma \ref{lem:characterization}.
\end{proof}

\section{Construction of one dimensional hull}\label{section:constructions}

By fixing the divisor $G$, we get the following result.
\begin{thm} Set $G=(n-2)P_\infty$. Let $a,b$ be in $\F_q[x]$ such that $\gcd (a,b)=1$. With the same notation as above, assume that 
\begin{enumerate}[(i)]
\item $supp((a)_0)\cap supp(D)=supp((b)_0)\cap supp(D)=\emptyset$,
\item $\deg a+\deg b=n-2,$
\item $(abh')(P_i),1\le i\le n,$ are nonzero squares in $\F_q.$
\end{enumerate}

Then $C_{\cal L}(D,(n-2)P_\infty-(a)_0+(z))$ is a $1$-$d$-hull code for any $z\in \F_q(x)$ satisfying $z(P_i)^2=(abh')(P_i)$ for all $1\le i\le n.$
\end{thm}

\begin{proof}
Under the above assumption, we have $2G-A-B-(h')-(n-2)P_\infty=G-A-B-(h')=G-(a)_0-(b)_0-(h')=\frac{1}{(abh')},$ where $A=(a)_0,B=(b)_0$, and $G-A-B=\frac{1}{(ab)}$ is obviously a principal divisor.
\end{proof}

\begin{rem}If $C$ is an $[n,1]$ code, then $C$ is a $1$-$d$-hull code if and only if C is self-orthogonal ($C\subset C^\perp$).
\end{rem}

From now on, we only consider $1$-$d$-hull code with parameters $[n,k]$, where $k>1.$

\begin{cor} \label{cor:even q}Assume that $q>4$ is even, $n\le q-2.$ Then there exists a $1$-$d$-hull MDS code with parameters $[n,n-s-1]$ for any $1\le s\le n-3.$
\end{cor}

\begin{proof}Let $\alpha,\beta\in \F_q$ such that $\alpha\not=\beta$ and the corresponding places $P_\alpha,P_\beta\notin supp(D).$ 
Set $a(x)=(x-\alpha)^{s},b(x)=(x-\beta)^{n-2-s}$. Since any element in $\F_q$ is a square, we get that $((abh')(\alpha_i))_{1\le i\le n}$ are nonzero square elements in $\F_q.$
\end{proof}

\begin{exam} For $q=2^3,n=4,5,6,$ and $s=1,$ using Magma \cite{Mag}, we give three $1$-$d$-hull MDS codes $C_1,C_2,C_3$ with parameters $[4,2,3],[5,3,3]$ and $[6,4,3]$, respectively as follows.

$$
C_1=\left(
\begin{array}{cccc}
1&0&w^3&w^3\\
0&1&w^4&w^3\\
\end{array}
\right),
C_2=\left(
\begin{array}{ccccc}
1&0&0&w^4&w^6\\
0&1&0&w&w^5\\
0&0&1&w^3&w^6\\
\end{array}
\right),
C_3=\left(
\begin{array}{cccccc}
1&0&0&0&w^3&w^2\\
0&1&0&0&1&w\\
0&0&1&0&w^6&w^2\\
0&0&0&1&1&w^2\\
\end{array}
\right).
$$

\end{exam}

\begin{cor} Assume that $q>5$ is odd, $n\le q-2$ and $(h'(P_i))_{1\le i\le n}$ are nonzero square elements in $\F_q.$ 
\begin{enumerate}
\item If  $n$ is even, then there exists a $1$-$d$-hull MDS code  with parameters $[n,n-2s-1]$ for any $1\le s\le \frac{n}{2}-2.$
\item If  $n$ is odd, $\frac{n-1}{2}\le q-1$ and $supp ((x)_0)\cap supp (D)=\emptyset$, then there exists a $1$-$d$-hull MDS code  with parameters $[n,n-2s-2]$ for any $0\le s\le \frac{n+1}{2}-3.$
\end{enumerate}
\end{cor}

\begin{proof}Let $\alpha,\beta\in \F_q$ such that $\alpha\not=\beta$ and the corresponding places $P_\alpha,P_\beta\notin supp(D).$ 
\begin{enumerate} 
\item Set $a(x)=(x-\alpha)^{2s},b(x)=(x-\beta)^{n-2-2s}$. It is easy to check that $((abh')(\alpha_i))_{1\le i\le n}$ are nonzero square elements in $\F_q.$
\item Set $a(x)=x(x-\alpha)^{2s},b(x)=(x-\beta)^{n-3-2s}$. For $1\le i\le n,$ take $\alpha_i$ from the set of square elements in $\F_q.$
\end{enumerate}
\end{proof}

Since the dual of an MDS code is again an MDS code, we derive the following

\begin{cor} Assume that $q>5$ is odd, $n\le q-2$ and $(h'(P_i))_{1\le i\le n}$ are nonzero square elements in $\F_q.$ 
\begin{enumerate}
\item If  $n$ is even, then there exists a $1$-$d$-hull MDS code  with parameters $[n,2s+1]$ for any $1\le s\le \frac{n}{2}-2.$
\item If  $n$ is odd, $\frac{n-1}{2}\le q-1$ and $supp ((x)_0)\cap supp (D)=\emptyset$, then there exists a $1$-$d$-hull MDS code  with parameters $[n,2s+2]$ for any $0\le s\le \frac{n+1}{2}-3.$
\end{enumerate}
\end{cor}

\begin{cor}\label{cor:square h}Let $q=p^m>5$ be odd and $ N\le q-2.$ Set
$$
K=
\begin{cases}
N-2s-1\text{ if $N$ is even and }1\le s\le \frac{N}{2}-2, \\
N-2s-2\text{ if $N$ is odd and }0 \le s\le \frac{N+1}{2}-3.\\
\end{cases}
$$ 
Then there exist $1$-$d$-hull MDS codes with parameters $[N,K]$ and $[N,N-K]$ with the following conditions:

\begin{enumerate}[(1)]
\item $p|N$, $(N-1)|(q-1),$ $N$ even, 
\item $q$  square, $(N-1)|(q-1)$, $N$ even, 

\item $ N|\frac{(q-1)}{2}$, 
\begin{enumerate}
\item $N$  even,
\item $q\equiv 1\pmod 4$, $N$ odd,
\end{enumerate}


\item $1\le r< m$, $\frac{n(p^r+1)}{2(p^r-1)}$ odd, $N=(t+1)n$, $n=\frac{q-1}{p^r+1}$ a square, $t$ odd, $1\le t\le \lfloor \frac{q-2}{n} \rfloor -1,$
\item $1\le r< m$, $\frac{n(p^r+1)}{2(p^r-1)}$ even, $N=(t+1)n$, $n=\frac{q-1}{p^r+1}$ a square, $1\le t\le \lfloor \frac{q-2}{n} \rfloor -1,$
\item $1\le r< m$, $N=(t+1)n$, $n=\frac{q-1}{p^r-1}$ a square, $r|\frac{m}{2}$, $1\le t\le p^r-2,$
\item $m=2m_0,q_0=p^{m_0}$, $t$ even, $1\le t< q_0$, $N=q_0t,$
\item  $m=2m_0,q_0=p^{m_0}$, $t$ odd, $1\le t< q_0$ and $1\le t\le q_0-1$ for $q=9$, $N=q_0t+1,$

\item $m$ even, $r=p^{m_0},{m_0}|\frac{m}{2},N=2tr^\ell,0\le \ell<m/{m_0},1\le t\le \min((r-1)/2 ,\lfloor \frac{q-2}{2r^\ell} \rfloor),$\\
\item $q\equiv 1\pmod 4, N=2p^\ell,0< \ell <m,$\\
\item $m$ even, $r=p^{m_0},{m_0}|\frac{m}{2},N=(2t+1)r^\ell+1,0\le \ell <m/{m_0},0 \le t \le \min((r-1)/2 ,\lfloor \frac{q-2-r^\ell}{2r^\ell} \rfloor),$\\
\item $q\equiv 1 \pmod 4 ,N=p^\ell+1,0< \ell< m.$\\
\end{enumerate}

\end{cor}
\begin{proof}Put $$h(x)=\prod\limits_{\alpha\in U}(x-\alpha).$$ For each case, it is enough to prove that, for any $\alpha\in U$, $h'(\alpha)$ is a nonzero square in $\F_q$. Take $U$ as follows.
\begin{itemize}
\item for 1), $U=\{\alpha\in \F_q|\alpha^N=\alpha\},$
\item for 2), $U=U_N=\{\alpha\in \F_q| \alpha^N=\alpha \},$
\item for 3), $U=U_N=\{\alpha\in \F_q| \alpha^N=1\},$
\item for  4)--6), take $U=U_{n}\cup \alpha_1U_{n}\cup \cdots \cup \alpha_tU_n$, $U_{n}=\{\alpha \in \F_q|\alpha^{n}=1\}$ and $\alpha_1,\hdots,\alpha_t\in \F_q\setminus U_{n}$.
\item for 7)--8), label the elements of $\F_{q_0}$ as $a_1,\hdots ,a_{q_0}. $ For some fixed element $\beta\in \F_q\setminus \F_{q_0}, $ take $U=\{a_k\beta+a_j|1\le k, j\le q_0\},$
\item for 9)--12),  label the element of $\F_r$ as $a_0,\hdots,a_{r-1}$, take $H$ as an $\F_r$-subspace and set $H_i=H+a_i \beta$ for some fixed element 
$\beta \in \F_q\setminus \F_r.$ Put $U=H_0\cup \cdots \cup H_{2t-1}$ or $U=H_0\cup \cdots \cup H_{2t}$. 
\end{itemize}

For 1)--3), it can be easily checked that $h'(\alpha)$ is a square for any $\alpha\in U.$ See also \cite[Theorem 2 and Theorem 4]{Sok}.

For 4)--6), it was already checked, in \cite[Theorem 3 and Theorem 4]{Sok-arxiv}, that $h'(\alpha)$ is a square for any $\alpha\in U$. Moreover, in \cite{Sok-arxiv}, $t\le p^r$ for 4)--5) and $t\le p^r-2$ for 6). Since $N=(t+1)n\le q-2,$ we take $t=\min (\lfloor \frac{q-2}{n} \rfloor -1,p^r)=\lfloor \frac{q-2}{n} \rfloor -1$ for 4)--5) and $t=\min (\lfloor \frac{q-2}{n} \rfloor -1,p^r-2)=p^r-2$ for 6).

For 7)--8), it was proved in \cite[Theorem 2]{Yan} that $h'(\alpha)$ is a square for any $\alpha \in U.$

For 9)--12),  it was proved in \cite[Theorem 4]{FangFu} that $h'(\alpha)$ is a square for any $\alpha \in U.$ Moreover, in \cite[Theorem 4]{FangFu}, $t\le (r-1)/2$ for 9). Since $N=2tr^\ell \le q-2,$ we take $\min((r-1)/2 ,\lfloor \frac{q-2}{2r^\ell} \rfloor).$ The range of $t$ and $\ell$ for 10)--12) follow from \cite[Theorem 4]{FangFu} with similar reasoning as 9).
\end{proof}

\begin{exam} From Corollary \ref{cor:square h} $(3)$ $a)$, for $q=3^4,n=8,s=1,$ we give a $1$-$d$-hull MDS code $C_5$ with parameters $[8,5,4]$ as follows.
$$
C_5=\left(
\begin{array}{cccccccc}
 1&0&0&0&0 &w^{49}& w^{57} &w^{35}\\
 0&1&0&0&0& w^{79}& w^{77}& w^{75}\\
 0&0&1&0&0& w^{24} &w^{52}&2\\
 0&0&0&1&0& w^{28}& w^{46} &w^{64}\\
 0&0&0&0&1& w^{34}& w^{72}&1\\
\end{array}
\right).
$$
\end{exam}

\begin{exam} From Corollary \ref{cor:square h} $(3)$ $b)$, for $q=19,n=9,s=1,$ we give a $1$-$d$-hull MDS code $C_4$ with parameters $[9,5,5]$ as follows.
$$
C_4=\left(
\begin{array}{ccccccccc}
 1&0&0&0&0&2 &13 &10&5\\
 0&1&0&0&0& 15&8&1&4\\
 0&0&1&0&0&1&7&6&8\\
 0&0&0&1&0& 16&2&7&8\\
 0&0&0&0&1&5& 12& 12&5\\
\end{array}
\right).
$$
\end{exam}


However when $(h'(P_i)_{1\le i\le n}$ are not all nonzero square elements in $\F_q$, we have the following.
\begin{thm} \label{thm:$1$-d-generalized} Let $q>5$ be an odd prime power. With the same notation as above, assume that $(h')_0=2(f)_0+(g_1)_0+(g_2)_0$ for some $f,g_1,g_2\in \F_q[x]$ such that $supp((g_1)_0) \cap supp((g_2)_0)=\emptyset.$ Assume further that $n+\deg (g_1)+\deg (g_2)\le q-1.$ 
Then for any $1\le s\le \lfloor \frac{(n-\deg (g_1)-\deg (g_2)}{2}\rfloor$, there exists a $1$-$d$-hull MDS code with parameters $[n,n-2s-\deg ((g_1)_0)+1].$
\end{thm}

\begin{proof} If we take $G=(n-2)P_\infty, A=2E+(g_1)_0$ for some divisor $E$ and $B=(g_2)_0+(n-2(s-1)-\deg (g_1)-\deg (g_2))P_\infty$, then plugging $G,A,B,(h')$ into (\ref{eq:char2}), we get
$$
\begin{array}{ll}
2G-A-B-(h')-(n-2)P_\infty&\\
=2(n-2)P_\infty-2E-(g_1)_0-(g_2)_0-(n-2(s-1)-\deg (g_1)-\deg (g_2))P_\infty&\\
~~-2(f)_0-(g_1)_0-(g_2)_0+(2\deg (f)+\deg (g_1)+\deg (g_2))P_\infty-(n-2)P_\infty&\\
=2\left(-E-(g_1)_0-(g_2)_0-(f)_0+(s-1+\deg (g_1)+\deg (g_2)+\deg(f))P_\infty\right)&\\
=\frac{1}{(efg_1g_2)^2}.&\\
\end{array}
$$

We now prove the existence of a place $E$ of degree $s-1\ge 0$ such that $E\notin {\cal S}:=\{P_\infty\}\cup supp(D)\cup supp((g_1)_0)\cup supp((g_2)_0).$ 
For $s=1$, we take $E=0.$ Then obviously $E \notin {\cal S}.$ Under the condition $n+\deg (g_1)+\deg(g_2)\le q-1$, there exists a zero place $F$ of degree one such that $F\notin {\cal S}.$ Take $E=(s-1)F$ with $s\ge 2.$

A simple calculation gives $$G-A-B=\left((2s-2)P_\infty-2E\right)+\left(\deg (g_1)+\deg (g_2)P_\infty-(g_1)_0-(g_2)_0\right)=\frac{1}{(e^2g_1g_2)},$$
which is obviously a principal divisor. The rest follows from Lemma \ref{lem:generalized}.

\end{proof}

\begin{cor}\label{con00:gen1}Let $q=p^m>5$ be odd, $n\le q-1,$ $p|n$ and $(n-1)|(q-1).$ Then there exist  $1$-$d$-hull MDS codes with parameters $[n,n-2s+1]$ and  $[n,2s-1]$ for $1\le s\le \lfloor \frac{n}{2}\rfloor.$

\end{cor}

\begin{proof} Set $h(x)=x^n-x.$ The derivative of $h(x)$ is $h'(x)=-1.$ 
Clearly, $h(x)$ has $n$ simple roots and it gives rise to $n$ distinct places of degree one. Set 
$
(f)_0=
(g_1)_0=
(g_2)_0=0.
$ 
Then $n+\deg (g_1)+\deg(g_2)=n\le (q-1).$ By applying Theorem \ref{thm:$1$-d-generalized}, the result follows.
\end{proof}

\begin{cor}\label{con0:gen1}Let $q=p^m>5$ be odd with $1\le r<m$ and $r|m$. Then there exist $1$-$d$-hull MDS codes with parameters $[p^r,p^r-2s+1]$ and  $[p^r,2s-1]$ for $1\le s\le \frac{p^r-1}{2}.$

\end{cor}

\begin{proof} Set $h(x)=x^{p^r}-x.$ The derivative of $h(x)$ is $h'(x)=-1.$ 
Clearly, $h(x)$ has $n=p^r$ simple roots and it gives rise to $n$ distinct places of degree one. Set 
$
(f)_0=
(g_1)_0=
(g_2)_0=0.
$ 
Then $n+\deg (g_1)+\deg(g_2)=p^r\le (q-1).$ By applying Theorem \ref{thm:$1$-d-generalized}, the result follows.
\end{proof}

\begin{cor}\label{con1:gen1}Let $q=p^m>5$ be odd with $m\ge 1$ and $n\le (q-2)$. Then 
\begin{enumerate}
\item there exist $1$-$d$-hull MDS codes with parameters $[n,n-2s+1]$ and $[n,2s-1]$ for $1\le s\le \lfloor (n-1)/2\rfloor,$
\item there exist $1$-$d$-hull MDS codes with parameters $[n,n-2s]$ and $[n,2s]$ for $n$ even and $1\le s\le n/2-1.$
\end{enumerate}
\end{cor}

\begin{proof} Set $h(x)=x^{n}-1.$ The derivative of $h(x)$ is $h'(x)=nx^{n-1}.$ Clearly, $h(x)$ has $n$ simple roots and it gives rise to $n$ distinct places of degree one.

Consider the following setting.
\begin{enumerate}
\item For $n$ odd, take
$
(f)_0=\frac{n-1}{2}(x)_0,
(g_1)_0=
(g_2)_0=0
$ and $1\le s\le \lfloor n/2\rfloor.$
For $n$ even, take
$
(f)_0=\frac{n-2}{2}(x)_0,
(g_1)_0=0,
(g_2)_0=(x)_0
$ and $1\le s\le \lfloor (n-1)/2\rfloor.$ Hence for both cases, $1\le s\le \lfloor (n-1)/2\rfloor.$

\item  For $n$ even, take 
$
(f)_0=\frac{n-2}{2}(x)_0,
(g_1)_0=(x)_0,
(g_2)_0=0.
$
\end{enumerate}

With the above setting, we get $n+\deg (g_1)+\deg(g_2)\le n+1\le (q-1).$ By applying Theorem \ref{thm:$1$-d-generalized}, the result follows.
\end{proof}

By considering additive cosets of some $\F_p$-subspaces of $\F_q$, we get the following construction.
\begin{cor}\label{con3:gen1}Let $q=p^m>5$ be odd with $m\ge 2$, $ r\le m-1$ and $r|m$. Set $N=(t+1)p^r$ with $\gcd (p,t+1)=1$ and $1\le t\le  \lfloor \frac{q-1-p^r}{2p^r} \rfloor$. Then there exist $1$-$d$-hull MDS codes with parameters $[N,N-2s+1]$ and $[N,2s-1]$ for $1\le s\le \lfloor p^r/2\rfloor .$
\end{cor}

\begin{proof} Let $U_0=\F_{p^r}$ and $U_i=(\beta_i+U_0)_{1\le i\le t}$ be $t$ nonzero distinct additive cosets of $U_0.$ Put $U=U_0\cup \left(\bigcup\limits_{i=1}^t\beta_i+U_0\right)$. Write

$$a(x)=(x^{p^r}-x), b(x)=\prod\limits_{i=0}^{t}(x-a(\beta_i)),h(x)=b(a(x)),$$
where $\beta_0=0.$ The derivative of $h(x)$ is $h'(x)=b'(a(x))a'(x)$, and thus $\deg (h')=(\deg (b)-1)\deg (a)=tp^r$ if $\gcd (t+1,p)=1.$
It is not difficult to check that $h(x)$ has all its $N$ simple roots in $U$, and it gives rise to $N$ distinct places of degree one. Setting $(f)_0=0,(g_1)=0$ and $(g_2)_0=(h')_0$, we get $N+\deg (g_1)+\deg (g_2)= (t+1)p^r+tp^r\le q-1$ if  $1\le t\le \lfloor \frac{q-p^r-1}{2p^r}\rfloor.$
By applying Theorem \ref{thm:$1$-d-generalized}, the result follows.
\end{proof}

By considering cosets of the multiplicative subgroup of $\F_q^*$ of order $n$, we get the following construction.
\begin{cor}\label{con2:gen1}Let $q=p^m>5$ be odd with $m\ge 1$ and $n|(q-1)$. Set $N=(t+1)n$ with $1\le t\le \lfloor \frac{q-n-2}{2n}\rfloor.$ Then there exist $1$-$d$-hull MDS codes with parameters as follows:
\begin{enumerate}
\item  $[N,n-2s+1]$ and $[N,tn+2s-1]$ for $1\le s\le \lfloor (n-1)/2\rfloor$, $p\not|(t+1),$
\item  $[N,2n-2s+1]$ and $[N,(t-1)n+2s-1]$ for $1\le s\le \lfloor (2n-1)/2\rfloor$, $p|(t+1),$
\item  $[N,n-2s]$ and $[N,tn+2s]$ for $n$ even, $1\le s\le \lfloor (n-1)/2\rfloor$, $p\not|(t+1),$
\item  $[N,2n-2s]$ and $[N,(t-1)n+2s]$ for $n$ even, $1\le s\le \lfloor (2n-1)/2\rfloor$, $p|(t+1),$
\item  $[N,N-2s+1]$ and $[N,2s-1]$ for $1\le s\le \lfloor (n-1)/2\rfloor$, $p\not|(t+1),$
\item  $[N,N-2s+1]$ and $[N,2s-1]$ for $1\le s\le \lfloor (2n-1)/2\rfloor$, $p|(t+1)$,
\item  $[N,N-2s]$ and $[N,2s]$ for $n$ even, $1\le s\le \lfloor (n-1)/2\rfloor$, $p\not|(t+1),$
\item  $[N,N-2s]$ and $[N,2s]$ for $n$ even, $1\le s\le \lfloor (2n-1)/2\rfloor$, $p|(t+1).$
\end{enumerate}
\end{cor}
\begin{proof} Let $U_n$ be a multiplicative subgroup of $\F_q^*$ of order $n.$
Let $\beta_1U_n,\hdots,\beta_tU_n$ be $t$ nonzero cosets of $U_n.$ Write

$$h(x)=(x^{n}-1)\prod\limits_{\lambda_1\in \beta_1U_n}(x-\lambda_1)\cdots \prod\limits_{\lambda_t\in \beta_tU_n}(x-\lambda_t).$$
The derivative of $h(x)$ is given by 
$$
h'(x)=nx^{n-1}\left(\sum\limits_{i=0}^t\prod\limits_{j=0,j\not=i}^t(x^n-\beta_j^n)\right),
$$
where $\beta_0=1.$ Clearly, $h(x)$ has $N=(t+1)n$ simple roots, and it gives rise to $N$ distinct places of degree one. Moreover, we have $(h')_0=(n-1)(x)_0+(g)_0$ for some $g\in \F_q[x]$ with the degree of $g$ equal to $tn$ if $p\not|(t+1)$ and equal to $(t-1)n$ if $p|(t+1) .$

For the proof of points $1)$ and $2)$, we have the following setting.
\begin{enumerate}[(a)]
\item If $n$ is odd, then we set 
$
(f)_0=\frac{n-1}{2}(x)_0,
(g_1)_0=(g)_0,
(g_2)_0=0.
$ 

\item If $n$ is even, then we set 
$
(f)_0=\frac{n-2}{2}(x)_0,
(g_1)_0=(g)_0,
(g_2)_0=(x)_0.
$
\end{enumerate}

For the proof of points $3)$ and $4)$, we set, for $n$ even,
$
(f)_0=\frac{n-2}{2}(x)_0,
(g_1)_0=(x)_0+(g)_0,
(g_2)_0=0.
$ 

For the proof of points $5)$ and $6)$ we have the following setting.
\begin{enumerate}[(a)]
\item If $n$ is odd, then we set 
$
(f)_0=\frac{n-1}{2}(x)_0,
(g_1)_0=0,
(g_2)_0=(g)_0.
$

\item If $n$ is even, then we set, 
$
(f)_0=\frac{n-2}{2}(x)_0,
(g_1)_0=0,
(g_2)_0=(x)_0+(g)_0.
$ 
\end{enumerate}

For the proof of points $7)$ and $8)$, we set, for $n$ even,
$
(f)_0=\frac{n-2}{2}(x)_0,
(g_1)_0=(x)_0,
(g_2)_0=(g)_0.
$

With the above setting, we get $N+\deg (g_1)+\deg (g_2)\le N+\deg (g)+1\le (t+1)n+tn+1\le q-1$ if  $1\le t\le \lfloor \frac{q-n-2}{2n}\rfloor.$
By applying Theorem \ref{thm:$1$-d-generalized}, the result follows.
\end{proof}
\begin{exam}Take $q=3^4,n=8$ and $t=1.$ By applying Corollary \ref{con2:gen1} 7) with $s=3$, we obtain a $1$-$d$-hull MDS code $C_6$ with parameters $[16,10,7]$, where its generator matrix is given by

{\scriptsize
$$
C_6=\left(
\begin{array}{ccccccc}
&w^{28}&w^{5}&w^{33}&w^{50}&w^{23}&w^{75}\\
&w^{41}&w^{35}&w^{47}&w^{44}&w^{18}&w^{75}\\
&w^{9}&w^{35}&w^{78}&w^{45}&w^{35}&w^{57}\\
&w^{68}&w^{31}&w^{12}&w^{4}&w^{30}&w^{8}\\
I_{10}&w^{66}&w^{39}&w^{73}&w^{67}&w^{8}&w^{46}\\
&w^{32}&w^{55}&2&w^{66}&w^{34}&w^{42}\\
&w^{53}&w&w^{17}&w^{74}&w^{33}&w^{72}\\
&w^{24}&w^{70}&w^{11}&w^{38}&w^{24}&w^{15}\\
&w^{11}&w^{25}&w^{9}&w^{16}&w^{77}&w^{70}\\
&w^{25}&w^{28}&w^{77}&2&w^{57}&w^{75}\\
\end{array}
\right).
$$
}
\end{exam}

\begin{exam}Take $q=3^4,n=8$ and $t=2.$ By applying Corollary \ref{con2:gen1} 8) with $s=2$, we obtain a $1$-$d$-hull MDS code $C_7$ with parameters $[24,20,5]$, where its generator matrix is given by
%

{\scriptsize
$$
C_7=\left(
\begin{array}{cccccccccccccccccccccccc}

&w^{2}&w^{28}&w^{44}&w^{77}\\
&w^{38}&w^{35}&w^{7}&w^{45}\\
&w^{64}&w^{62}&w^{76}&w^{79}\\
&2&w^{31}&w^{49}&w^{29}\\
&w^{8}&w^{59}&w^{65}&w^{36}\\
&w^{48}&w^{5}&w^{6}&w^{45}\\
&w^{23}&w^{44}&w^{68}&w^{69}\\
&w^{22}&w^{35}&w^{22}&w^{41}\\
&w^{29}&w^{49}&w^{28}&w^{17}\\
I_{20}&w^{34}&w^{41}&w&w^{21}\\
&w^{10}&w^{76}&w^{39}&w^{11}\\
&w^{46}&w^{11}&w^{34}&w^{8}\\
&w^{77}&w^{28}&w^{41}&w^{41}\\
&w&w^{2}&w^{7}&w^{14}\\
&w^{37}&w^{56}&w^{24}&w^{23}\\
&w^{65}&w^{46}&w^{37}&w^{20}\\
&w^{64}&w^{33}&w^{19}&w^{68}\\
&w^{37}&w^{77}&w^{51}&1\\
&w^{37}&w^{11}&w^{55}&w^{47}\\
&w^{27}&w^{2}&w^{28}&w^{68}\\
\end{array}
\right).
$$
}
\end{exam}


\begin{rem} From Corollary \ref{con2:gen1}, one may get longer $1$-$d$-hull MDS codes by considering larger $t,$ that is, $t\ge \lfloor \frac{q-n-2}{2n}\rfloor+1$. In this situation, the existence of such codes depends on the existence of a zero place $F$ of degree one such that $F\notin supp(D)\cup supp((h')_0),$ which depends on the number of zeros of the derivative $h'(x).$ For example, applying $q=3^4,n=8$ to Corollary \ref{con2:gen1} 7)--8), the value of $t$ lies between $0$ and $4$, and we get $1$-$d$-hull MDS codes of lengths $N=(t+1)n$. However, for $t=5$, using Magma \cite{Mag}, we can find a zero place $F$ of degree one such that $F\notin supp(D)\cup supp((h')_0)$, and thus we obtain $1$-$d$-hull MDS codes with parameters $[48, 46,3],[48, 44,5],[48, 42,7],[48, 40,9],[48, 38,11],$ $[48, 36,13],[48, 34,15].$
\end{rem}

\section{Conclusion}\label{sec:conclusion}
In this paper, we deal with linear codes having one dimensional hull. The hull is defined with respect to Euclidean inner product, and we construct families of $1$-$d$-hull MDS codes from algebraic geometry codes of genus zero. For the future work, with the same spirit, it is worth considering the hull with respect to Hermitian inner product on the one hand, and consider codes from higher genus and with higher dimensional hull on the other hand.
%
\medskip
%


\begin{thebibliography}{99}
\bibitem{AssKey} E. F. Assmus, Jr and J. D. Key, ``Affine and projective planes," {\em Discrete Math.} vol. 83, pp. 161--187, 1990.
\bibitem{Ball} S. Ball, ``On sets of vectors of a finite space in which every subset of basis size is a basis," {\em J. Eur. Soc.} 14, pp. 733-748, 2012.
\bibitem{Mag} W. Bosma and J. Cannon, {\em   Handbook of Magma Functions}, Sydney, 1995.
\bibitem{CarGui} C. Carlet and S. Guilley, ``Complementary dual codes for counter-measures to side-channel attacks," In:E.R.
Pinto {\em et al.} (eds.), {\em Coding Theory and Applications,} CIM Series in Mathematical Sciences, vol. 3, pp.
97--105, Springer (2014), {\em J. Adv. Math. Commun.} 10(1), pp.131--150, 2016.
\bibitem{CarGunOzbOzkSol} C. Carlet, C. G\"{u}neri, F. \"{O}zbudak, B. \"{O}zkaya and P. Sol\'e, ``On linear complementary pairs of codes," {\em IEEE
Trans. Inf. Theory,} 64(10), pp. 6583--6589, 2018.
\bibitem{CarLiMes} C. Carlet, C. Li and S. Mesnager, ``Linear codes with small hulls in semi-primitive case," {\em Des. Codes Cryptogr.} https://doi.org/10.1007/s10623-019-00663-4
\bibitem{CarMesTanQiPel18} C. Carlet, S. Mesnager, C. Tang, Y. Qi and R. Pellikaan, `` Linear codes over $\F_q$ are equivalent to LCD codes for $q>3,$" {\em IEEE Trans. Inf. Theory,} 64(4), pp. 3010--3017, 2018.
\bibitem{CarMesTanQi19} C. Carlet C, S. Mesnager, C. Tang and Y. Qi, ``New characterization and parametrization of LCD codes," {\em IEEE Trans. Inf. Theory,} 65(1), pp. 39--49, 2019.
\bibitem{CarMesTanQi18-2} C. Carlet, S. Mesnager, C. Tang and Y. Qi, `` Euclidean and Hermitian LCD MDS codes," {\em Des. Codes Cryptogr.} 86, pp. 2605--2618, 2018.
\bibitem{CarMesTanQi19-2} C. Carlet, S. Mesnager, C. Tang and Y. Qi, ``On $\sigma$-LCD codes," {\em IEEE Trans. Inf. Theory,} 65(3), pp. 1694--1704, 2019.
\bibitem{ChenLiu}B. Chen and H. Liu, ``New constructions of MDS codes with complementary duals," {\em IEEE Trans. Inf. Theory,} 64(8), pp. 5776--5782, 2018.

\bibitem{FangFu} W. Fang and F. Fu, ``New Constructions of MDS Euclidean Self-dual Codes from GRS Codes and Extended GRS Codes," {\em IEEE Trans. Inform. Theory}, vol. 65(9), pp. 5574--5579, 2019.

\bibitem{GraGul} M. Grassl and T. A. Gulliver, ``On Self-Dual MDS Codes" {\em ISIT 2008}, Toronto, Canada, July 6 --11, 2008
\bibitem{Gue} K. Guenda, ``New MDS self-dual codes over finite fields,"  {\em Des. Codes Cryptogr.} 62, pp. 31--42, 2012.

\bibitem{Jin} L. Jin, ``Construction of MDS codes with complementary duals," {\em IEEE Trans. Inf. Theory,} 63(5), pp. 2843--2847, 2017.
\bibitem{JinBee}L. Jin and P. Beelen, ``Explicit MDS Codes With Complementary Duals," {\em IEEE Trans. Inf. Theory,} 64(11), pp. 7188 --7193, 2018.

\bibitem{JinXin} L. Jin and C. Xing, ``New MDS self-dual codes from generalized Reed-Solomon codes," {\em IEEE Trans. Inform. Theory}, vol. 63(3) , pp. 1434 --1438, 2017.

\bibitem{Leon82}J. Leon, ``Computing automorphism groups of error-correcting codes," {\em IEEE Trans. Inf. Theory,} 28(3), pp. 496--511, 1982.
\bibitem{Leon91} J. Leon, ``Permutation group algorithms based on partition, I: Theory and algorithms," {\em J. Symb. Comput.} 12, pp. 533--583, 1991.
\bibitem{LiDingLi} C. Li, C. Ding and S. Li, ``LCD cyclic codes over finite fields," {\em IEEE Trans. Inf. Theory,} 63(7), pp. 4344--4356, 2017.
\bibitem{LiLiDingLiu}S. Li, C. Li, C. Ding and H. Liu, ``Two Families of LCD BCH codes," {\em IEEE Trans. Inf. Theory,} 63(9), pp. 5699--5717, 2017.
\bibitem{LiZeng} C. Li and P. Zeng, ``Constructions of linear codes with one-dimensional hull," {\em IEEETrans. Inf. Theory,} 65 (3), pp. 1668--1676, 2019.

\bibitem{MacSlo} F. J. MacWilliams and N. J. A. Sloane, {\em The Theory of Error-Correcting
Codes,} Amsterdam, The Netherlands: North Holland, 1977.

\bibitem{Massey}J. L. Massey, ``Linear codes with complementary duals," {\em Discret. Math.}106(107),337--342 (1992).
\bibitem{MesTanQi}S. Mesnager, C. Tang and Y. Qi, ``Complementary dual algebraic geometry codes," {\em IEEE Trans. Inf. Theory,} 64(4), pp. 2390--2397, 2018.

\bibitem{SangJitLingUdom}E. Sangwisut, S. Jitman, S. Ling, and P. Udomkavanich, ``Hulls of cyclic and negacyclic codes over finite fields," {\em Finite Fields Appl.} 33, pp. 232--257, 2015.
\bibitem{Sendrier97}N. Sendrier, ``On the dimension of the hull," {\em SIAM J. Discret. Math.} 10(2), pp. 282--293, 1997.

\bibitem{Sendrier00}N. Sendrier, ``Finding the permutation between equivalent codes: the support splitting algorithm," {\em IEEE Trans. Inf. Theory,} 46(4), pp. 1193--1203, 2000.

\bibitem{ShiYueYan} X. Shi, Q. Yue and S. Yang, ``New LCD MDS codes constructed from generalized Reed-Solomon codes," {\em J. Algebra Appl.} 18, 1950150, 2018.

\bibitem{Ske}G. Skersys, ``The average dimension of the hull of cyclic codes," {\em Discret. Appl. Math.} 128(1), pp. 275--292, 2003.

\bibitem{Sok} Lin Sok, ``Explicit constructions of MDS self-dual codes," {\it IEEE Transactions on Information Theory,} vol. 66(6), pp. 3603--3615, 2020, doi:10.1109/TIT.2019.2954877
\bibitem{Sok-arxiv} Lin Sok, ``New families of self-dual codes," https://arxiv.org/abs/2005.00726
\bibitem {Stich} H. Stichtenoth, ``Algebraic function fields and codes," Springer, 2008.

\bibitem{TongWang} H. Tong and X. Wang, ``New MDS Euclidean and Hermitian self-dual codes over finite fields," {\em Adv. in Pure Math.} vol. 7, pp. 325--333, 2017.
\bibitem{Yan} H. Yan, ``A note on the constructions of MDS self-dual codes," {\em Cryptogr. Commun.} 11, pp. 259--268, 2019.
https://doi.org/10.1007/s12095-018-0288-3

\bibitem{YanLiuLiYang}H. Yan, H. Liu, C. Li and S. Yang, ``Parameters of LCD BCH codes with two lengths," {\em Adv. Math. Commun.}12(3), pp. 579--594, 2018.
\bibitem{YangMassey} X. Yang and J. L. Massey, ``The condition for a cyclic code to have a complementary dual," {\em Discret. Math.} 126(1--3), pp. 391--393, 1994.






















\end{thebibliography}
\end{document}